\documentclass{iucr}              % DO NOT DELETE THIS LINE

     \journalcode{J}              % Indicate the journal to which submitted
\usepackage{url}
\usepackage{graphicx}
\usepackage{tabularx}
\usepackage{amsmath,amstext,amssymb}
\usepackage{xspace}

\usepackage{pifont}% http://ctan.org/pkg/pifont
\newcommand{\cmark}{\ding{51}}%

\newcommand{\mcrot}[4]{\multicolumn{#1}{#2}{\rlap{\rotatebox{#3}{#4}~}}}
\newcommand{\vectorspace}[2]{\ensuremath{\mathbb{#1}^{#2}}\xspace}
\newcommand{\point}[1]{\ensuremath{p_{#1}}}
\newcommand{\set}[1]{\ensuremath{\mathbf{#1}}}
\newcommand{\pointset}{\ensuremath{\set{P} = \{\point{1}, \point{2}, \ldots, \point{N} \}}}
\newcommand{\sothree}{\ensuremath{SO(3)}\xspace}
\newcommand{\dilog}[1]{\ensuremath{\operatorname{Li}_2\left({#1}\right)}\xspace}
\newcommand{\radius}{\ensuremath{\theta}\xspace}
\newcommand{\Bstrut}[1]{\rule[- #1 ex]{0pt}{0pt}}   % = `bottom' strut

\usepackage{amsthm}
\newtheorem*{theorem*}{Theorem}%[Section]

\newcommand{\shortcite}[1]{\cite{#1}}

\begin{document}                  % DO NOT DELETE THIS LINE

     %-------------------------------------------------------------------------
     % The introductory (header) part of the paper
     %-------------------------------------------------------------------------

     % The title of the paper. Use \shorttitle to indicate an abbreviated title
     % for use in running heads (you will need to uncomment it).

%\title{Improving the Performance of Forward Modelling Based Pattern Indexing by Optimal Selection of Orientations}
%\shorttitle{Optimal Selection of Orientations}
\title{Improved Orientation Sampling for Indexing Diffraction Patterns of Polycrystalline Materials}
\shorttitle{Improved Orientation Sampling}

     % Authors' names and addresses. Use \cauthor for the main (contact) author.
     % Use \author for all other authors. Use \aff for authors' affiliations.
     % Use lower-case letters in square brackets to link authors to their
     % affiliations; if there is only one affiliation address, remove the [a].

\cauthor{Peter Mahler}{Larsen}{pmla@fysik.dtu.dk}{}
%\author{Peter Mahler}{Larsen}
\author{S\o ren}{Schmidt}

\aff{Department of Physics, Technical University of Denmark, 2800 Kongens Lyngby, \country{Denmark}}

     % Use \shortauthor to indicate an abbreviated author list for use in
     % running heads (you will need to uncomment it).

%\shortauthor{Soape, Author and Doe}

     % Use \vita if required to give biographical details (for authors of
     % invited review papers only). Uncomment it.

%\vita{Author's biography}

     % Keywords (required for Journal of Synchrotron Radiation only)
     % Use the \keyword macro for each word or phrase, e.g. 
     % \keyword{X-ray diffraction}\keyword{muscle}

%\keyword{keyword}

     % PDB and NDB reference codes for structures referenced in the article and
     % deposited with the Protein Data Bank and Nucleic Acids Database (Acta
     % Crystallographica Section D). Repeat for each separate structure e.g
     % \PDBref[dethiobiotin synthetase]{1byi} \NDBref[d(G$_4$CGC$_4$)]{ad0002}

%\PDBref[optional name]{refcode}
%\NDBref[optional name]{refcode}

\maketitle                        % DO NOT DELETE THIS LINE

\begin{synopsis}
A method for generating high-quality discretizations of \sothree is described and compared with existing methods.
\end{synopsis}

\begin{abstract}
Orientation mapping is a widely used technique for revealing the microstructure of a polycrystalline sample.
The crystalline orientation at each point in the sample is determined by analysis of the diffraction pattern, a process known as \emph{pattern indexing}.
A recent development in pattern indexing is the use of a brute-force approach, whereby diffraction patterns are simulated for a large number of crystalline orientations, and compared against the experimentally observed diffraction pattern in order to determine the most likely orientation.  Whilst this method can robust identify orientations in the presence of noise, it has very high computational requirements.
In this article, the computational burden is reduced by developing a method for nearly-optimal sampling of orientations.
By using the quaternion representation of orientations, it is shown that the optimal sampling problem is equivalent to that of optimally distributing points on a four-dimensional sphere.  In doing so, the number of orientation samples needed to achieve a indexing desired accuracy is significantly reduced.
Orientation sets at a range of sizes are generated in this way for all Laue groups, and are made available online for easy use.
\end{abstract}

     %-------------------------------------------------------------------------
     % The main body of the paper
     %-------------------------------------------------------------------------
     % Now enter the text of the document in multiple \section's, \subsection's
     % and \subsubsection's as required.

\section{Introduction}
\label{sec:introduction}

In many types of diffraction experiments, the aim is to determine the orientation of the diffracted crystallite volume which creates the experimentally observed pattern.  For example, when studying a multigrain sample with the 3DXRD technique~\cite{poulsen2001three}, a `grain map' is constructured by finding the crystalline orientation at each point in the sample.
The process of determining the crystalline orientation from a diffraction pattern on the detector is known as \emph{pattern indexing}.  Throughout this article, we use the term `experimental pattern' to denote an image of a diffraction pattern as recorded on a detector.% using monochromatic radiation.

The most widely-used pattern indexing methods work `backwards' from features in the observed data to an orientation.  Such methods are typically highly efficient, but can fail in the presence of noise.
%Whilst there exist efficient methods for pattern indexing in a range of modalities, accurate determination of the orientation can be challenging in the presence of noise.
A well-known example is in Electron Backscatter Diffraction (EBSD), where the Hough transform is used to find lines in the backscattered Kikuchi pattern, from which the orientation can be determined~\cite{adams1993orientation}.  Under noisy conditions, however, the Kichuchi lines can no longer be reliably identified and the indexing process fails as a consequence.

The desire to analyze diffraction patterns under less-than-ideal conditions has motivated the development of forward modelling based pattern indexing, also known as dictionary-based indexing.
%
%Traditional pattern indexing methods determine the orientation by analyzing features in the data, such as using the Hough transform to find Kikuchi lines in Electron Backscatter Diffraction (EBSD) data.  
%In recent years, the increased computational power has enabled the practical implementation of a new approach to pattern indexing, called dictionary
%highly robust approach for pattern indexing, namely forward modelling, also known as dictionary-based indexing. 
%
%Rather than attempt to determine the orientation by analyzing features in the data, such as, for example, using the Hough transform to find Kikuchi lines in Electron Backscatter Diffraction (EBSD) data, the experimental data is compared against a precomputed set of simulated patterns (the dictionary) with known orientations.  To a first approximation, the orientation of the experimental pattern is determined by finding the dictionary pattern with the highest similarity, where the similarity is a distance function of the simulated and experimentally observed intensities.
%
In a forward model, rather than working backwards from the data, the orientation is found using a brute-force approach.
A dictionary is constructed by selecting a set of orientations, and generating \emph{simulated} patterns for each of them.
A requirement for simulating patterns is that the crystal phase is known \emph{a priori}, or, if indexing a multiphase material, that the set of candidate phases is known.
%Given an experimental pattern, 
%Each experimental pattern is compared against a precomputed set of simulated patterns (the dictionary) with known orientations.
%The orientation of an experimental pattern is then determined by finding the dictionary pattern with the highest similarity, where the similarity is a distance function of the simulated and experimentally observed intensities. 

To index an experimental pattern, it is compared against every simulated pattern in the dictionary, and the dictionary pattern with the highest similarity determines the orientation (in a multiphase material this also determines the phase).
Here, the similarity is determined by the difference in the pixel intensities in the simulated and experimental patterns.  By using the full image information (i.e.~all pixel intensities), the similarity exhibits a continuous degradation with increasing noise, as opposed to the catastrophic degradation exhibited when looking for specific features in the experimental pattern.
%Here, the similarity is a distance function of the simulated and experimentally observed pixel intensities.
%By using the full image information instead of looking for specific features, this approach is highly robust to noise.

A significant drawback of the forward modelling approach, however, is the computational effort required: each experimental pattern must be tested against every dictionary pattern.  Since, the accuracy of the pattern indexing process depends on the granularity of the set of dictionary orientations, a more accurate indexing requires a larger set.  Increasing the number of dictionary orientations, however, increases the time required to index a pattern.  Since the objectives of increased accuracy and reduced running time are in opposition to each other, we ask the question: how can we achieve the highest accuracy with the fewest dictionary orientations?  In this article, we describe a method for doing so with the use of quaternions.

%A pertinent question is then: on what basis should the dictionary orientations be selected?  There is arguably no canonical answer, since any solution involves a trade-off between multiple desirable yet irreconcilable properties.
%Nonetheless, in this article we provide a new answer, by choosing orientations so as to minimize the maximum error of a dictionary.

\subsection{Measurement of Dictionary Orientation Sets}
Orientations can be conveniently represented using unit quaternions~\cite{altmann2005rotations}.
%, which are vectors of the form $q = \{w, x, y, z\}$ where $w, x, y$ and $z$ are real quantites such that $|\vec{q}| = 1$.  A rotation matrix can be derived from a quaternion by:
Briefly, a quaternion is a four dimensional vector of the form $q = \{ w, ix, jy, kz \}$, where $w$, $x$, $y$ and $z$ are real numbers, and $i$, $j$ and $k$ are imaginary numbers which generalize the better-known complex numbers.
Unit quaternions represent points on a four-dimensional hypersphere, a space formally known as \vectorspace{S}{3} and which consists of all vectors which satisfy $\sqrt{w^2 + x^2 + y^2 + z^2} = 1$.  This space is a double covering of \sothree, the group of rotations in three-dimensional Euclidean space $\left( \vectorspace{R}{3} \right)$.
The double covering relationship means that $-q$ and $q$ represent the same orientation, which is evident when considering the quaternion-derived rotation matrix:
\begin{equation}
\set{U}_q =
\begin{bmatrix}
%w^2 + x^2 - y^2 - z^2
1 - 2y^2 - 2z^2
&2xy - 2wz
&2xz + 2wy\\
2xy + 2wz
%&w^2 - x^2 + y^2 - z^2
&1 - 2x^2 - 2z^2
&2yz - 2wx\\
2xz - 2wy
&2yz + 2wx
%&w^2 - x^2 - y^2 + z^2
&1 - 2x^2 - 2y^2
\end{bmatrix}
\end{equation}
It can be seen that in each element of $\set{U}_q$, the sign of the quaternion cancels out.
%
%Since every element of $\set{U}_q$ contains only products of exactly two elements of $\vec{q}$, the sign of $\vec{q}$ cancels out.
%
By using the quaternion representation, the problem of selecting an optimal set of dictionary orientations is equivalent to finding an optimal distribution of a set of points on \vectorspace{S}{3}.  To do so, we must first decide what constitutes a good distribution.

The misorientation between two orientations in quaternion form, $p$ and $q$, is given by:
\begin{equation}
\alpha \left(p, q \right) = 2 \arccos |\langle p, q \rangle|
\end{equation}
where $\langle p, q \rangle$ denotes the inner product of $p$ and $q$.
In many previous studies, dictionary orientation sets are quantified by the misorientation between neighbouring orientations, for example, the average value of $\alpha \left(p, q \right)$ over all pairs of nearest neighbours $p$ and $q$.  This may be  adequate when the orientation set has a known, grid-like structure, but it does not constitute a universal measure of quality.
To illustrate this with a pathological example, consider an orientation set, \set{Q}, where all orientations lie at the same point.  The misorientation between all pairs of orientations is zero, that is
\begin{equation}
\alpha(p, q) = 0 \hspace{4mm} \forall p \in \set{Q}, q \in \set{Q}
\end{equation}
yet the set constitutes the worst possible dictionary.  A good measure of quality should instead consider the misorientation between the dictionary set and any possible experimental orientation.  We define the error term as the maximum misorientation between these two, i.e. how far can an experimental orientation lie from the dictionary?  More specifically, this error term is given by:
\begin{equation}
\alpha_{\text{max}} = \max_{x \in \sothree} \min_{q \in \set{Q}} \alpha \left( x, q \right)
\end{equation}
This quantity can be minimized by solving the \emph{spherical covering problem} in \vectorspace{S}{3}.  Given $N$ hyperspherical caps of equal radius, $r$, called the \emph{covering radius}, the spherical covering problem asks how to arrange the caps to cover the surface of \vectorspace{S}{3} with minimal $r$.  We describe this problem in detail in Section~\ref{sec:theory}.

%As we show later, this error is equal to twice the \emph{covering radius} in \vectorspace{S}{3}.  Finding a point set on a hypersphere with a small covering radius is known as the spherical covering problem; given $N$ hyperspherical caps of equal radius, $r$, the spherical covering problem asks how to arrange the caps to cover \vectorspace{S}{3} so as to minimize $r$ (the covering radius).  This is described in detail in Section~\ref{sec:theory}.

By creating orientation sets with a small covering radius, we can either reduce the number of orientations required to achieve a desired error tolerance (thereby reducing the running time of forward modelling pattern indexing), or simply improve the error distribution for a fixed number of orientations.  Creation of such sets is the principal contribution of this work.

\subsection{Previous Work}
Forward modelling has been successfully applied in many types of diffraction-based experiments, including the indexing of 3D X-ray diffraction microscopy data~\cite{Li2013,schmidt2014grainspotter}, EBSD data~\cite{chen2015dictionary} and electron channeling patterns~\cite{singh2017dictionary}.  Any forward modelling method requires a discretization of \sothree.  Whilst many such discretization methods have been developed, here we consider only three which are both successful and commonly used amongst crystallographers.
%forward modelling method requires accurate models of the scattering process and of the detector (in order that a simulated pattern is a high-fidelity representation of an experimental pattern with the same orientation), and a discretization of \sothree.  In this work we concern ourselves only with the latter.

Yershova et al.~\shortcite{Yershova2010} have developed an incremental infinite sequence based on the Hopf fibration.  The method generates orientations deterministically, with proven maximal dispersion reduction when used as a sequence.  Furthermore, the orientation sets are isolatitudinal, which permits expansion into spherical harmonics~\cite{dahms1989iterative}, refinable, and can be generated on-the-fly.  Whilst the method has many desirable properties, it is developed for the purpose of robot motion planning and is not easily integrated with crystallographic fundamental zones.  To remedy this,  Ro\c{s}ca et al.~\shortcite{rosca2014new} have developed `cubochoric' coordinates, in which an area-preserving Lambert projection is used to map points from a cubic grid onto any desired crystallographic fundamental zone in \sothree.  %The resulting point sets are \emph{isochoric} (they have equal volume), as well as being refinable and fast to generate.
A different approach, developed by Karney~\shortcite{Karney2007} for use in molecular modelling, is to generate sets which attempt to solve the spherical covering problem.  Inspired by the observation that body-centred cubic (BCC) grids solve the covering problem in \vectorspace{R}{3}, BCC grids are constructed in Rodrigues-Frank (RF) space~\cite{frank1988rf,morawiec1996rodrigues} in order to generate good coverings in \sothree.

%
%table:method_comparison was here
\begin{table}%[h]
\hcaption{Summary of properties of different methods of orientation set generation.  Existing methods prioritize fast generation and a grid-like structure.  In our work we optimize the covering radius at the expense of all other properties.  The optimality gap for a set of $N$ orientations is the percentage difference of its covering radius to that of the simplex bound (c.f.\ Section~\ref{sec:lowerbound}).
$^{1}$Non-isolatitudinal sets do not permit an expansion into spherical harmonics, though any orientation set can be expanded into hyperspherical harmonics~\cite{mason2008hyperspherical,mason2009relationship}.
$^{2}$These orientation sets can be mapped out into 7 of 11 Laue group fundamental zones (c.f.\ Section~\ref{sec:theory_symmetry}).}
\label{table:method_comparison}
\begin{tabular}{lcccccc}
%\hline
\multicolumn{1}{l}{\textbf{Method}} & \mcrot{1}{c}{60}{Fast generation} & \mcrot{1}{c}{60}{Refinable} & \mcrot{1}{c}{60}{Isolatitudinal} & \mcrot{1}{c}{60}{Isochoric}  & \mcrot{1}{c}{60}{Crystallographic} & \mcrot{1}{c}{60}{Opt. gap at $N \approx 10^5$}\\
% & \mcrot{1}{c}{60}{$N$ required for $2\theta<1^\circ$} \\
\hline
Random sampling				& \cmark 	& \cmark & -$^1$	& -		& \cmark		& $127\%$\\
Hopf fibration				& \cmark 	& \cmark & \cmark		& -		& -		& $59.9\%$\\
Cubochoric					& \cmark 	& \cmark & \cmark		& \cmark	& \cmark				& $40.8\%$\\
Octahedral BCC				& \cmark 	& \cmark & -$^1$	& -		& -$^2$		& $15.4\%$\\% & $1.42 \times 10^7$\\ %7087992*2
Present work					& - 			& -		 & -$^1$	& -		& \cmark	& $4.64\%$\\%& $1.21 \times 10^7$\\
%\hline
\end{tabular}
\end{table}

Table~\ref{table:method_comparison} summarizes the properties of the different methods of generating orientation sets.
Each of the three aforementioned techniques attempts to solve slightly different problems and involves different trade-offs as a consequence, although one feature they have in common is fast generation.  We take an alternative approach, sacrificing other properties in pursuit of creating the `best' possible orientation sets.  Whilst this approach requires a significant up-front computational effort, this is a good trade-off when the resulting sets will subsequently be used many times.  We emphasize that whilst the orientational error is critical to forward modelling, there are many other sources of error in any modality (see Ram et al.~\shortcite{ram2017error} for a comprehensive analysis in an EBSD context).

The rest of this article is organized as follows: in Section~\ref{sec:theory} we define the spherical covering problem on \vectorspace{S}{d}, show how this relates to the problem of finding an optimal set of orientations,
and derive a conjectured lower bound.  We describe the generation of orientation sets in Section~\ref{sec:method}.  Results on the covering radius and error distributions of the resulting orientation are given in Section~\ref{sec:results}.  Lastly, the advantages and drawbacks of the method presented are discussed in Section~\ref{sec:summary}.

\newcommand{\circumcentre}{\ensuremath{X}\xspace}
\section{Error Quantification of Orientation Sets}
\label{sec:theory}

In order to compare different orientations sets we must define a measure of quality.  Here, we describe the covering radius of a set, which we argue is the canonical error measure since it determines the maximum possible error.
We will first describe the sphere covering problem for Euclidean and spherical geometries, and then show that the problem of generating optimal orientation sets is a special case of the spherical covering problem.

\subsection{Spherical Coverings}

The sphere covering problem is best known in Euclidean geometries.  In $\vectorspace{R}{d}$, it asks `for the most economical way to cover $d$-dimensional space with equal overlapping spheres'~\cite{conway1998sphere}.  Optimal coverings are known for $d=1$ and $d=2$, which are equally spaced points on a line and a hexagonal lattice, respectively, and optimal \textit{lattice coverings} are known for $1 \leq d \leq 5$.
%but for $3 \leq d \leq 5$ the existence of better non-lattice (aperiodic) coverings has not been disproven. %also from conway1998sphere.

The presence of curvature in spherical geometries renders the covering problem vastly more challenging.  In \vectorspace{S}{d}, the spherical covering problem asks for the most economical way to cover the surface of $\vectorspace{S}{d}$ with equal overlapping hyperspherical caps.  In \vectorspace{S}{1}, the optimal covering is a set of $N$ points with angle $\frac{2\pi}{N}$ between adjacent points.  For $d > 1$, however, there is no general formula for determining the optimal spherical covering.  Furthermore, unlike in \vectorspace{R}{d}, the configuration of the optimal covering depends on the number of points in the covering.  For example, for $d=2$, the known optimal configurations are the vertices of the tetrahedron, the octahedron and the icosahedron.  Hardin et al.~have found \emph{putatively} optimal coverings~\cite{SloaneCoverings} for $d=2$ at other values of $N$, but these have been found using numerical optimization and are not provably optimal.
%\textbf{todo: add citations for known results}

\subsection{Covering Radius and Covering Density}

For coverings on \vectorspace{S}{d}, the two (equivalent) measures of quality are the covering radius and the covering density.  Given a discrete collection of points $\pointset \in \vectorspace{S}{d}$, the covering radius, $\radius$, is defined as the largest angular distance between any point in \vectorspace{S}{d} and \set{P}, that is
\begin{equation}
\radius = \max\limits_{x \in \vectorspace{S}{d}} \min\limits_{\point{} \in \set{P}} \arccos \langle x, \point{} \rangle
\label{eq:covering_radius}
\end{equation}
where $\langle x, \point{} \rangle$ denotes the inner product of $x$ and $\point{}$.  Then, \set{P} covers the surface of \vectorspace{S}{d} with $N = |\set{P}|$ equal hyperspherical caps of radius $\radius$.  The covering density,  $\tau_{d}(\radius)$, is given by ratio of the sum of the surface area of the caps to the surface area of unit $d$-sphere,
\begin{equation}
\tau_{d}(\radius) = N \frac{C_d(\radius)}{S_d(1)}
\label{eq:spherical_density}
\end{equation}
where
\begin{equation}
C_{d}(\radius) = \int\limits_{0}^{\tan(\radius)} \frac{S_{d-1}(r)}{(1 + r^2)^2} dr,
\hspace{4mm}
S_{d-1}(\radius) = \frac{d\pi^{d/2}}{\Gamma \left(\frac{d}{2} + 1 \right)}\radius^{d-1}
\label{eq:spherical_density2}
\end{equation}
where $S_{d-1}(\radius)$ is the surface area of the $d$-sphere of radius $\theta$ and $C_{d}(\radius)$ is the surface area of a hyperspherical cap of radius $\theta$ (c.f. Appendix \ref{app:simplex_bound_derivation} for derivation).
%\textbf{todo: verify this for 2 and 3 dimensions}\\
%
To find the covering radius, we need to determine the Voronoi cell of each point $\point{i} \in \set{P}$.  The Voronoi cell of point $\point{i}$, denoted $\text{Vor}(\point{i})$, consists of all points of $\vectorspace{S}{d}$ that are at least as close to $\point{i}$ as to any other $\point{j}$.  More specifically:
\begin{equation}
\text{Vor}(\point{i}) = \{x \in \vectorspace{S}{d} \mid \arccos \langle x, \point{i} \rangle \leq \arccos \langle x, \point{j} \rangle \;\; \forall j \}
\label{eq:voronoi}
\end{equation}
Since the vertices of the Voronoi cells are the points which locally maximize the angular distance from \set{P}, the covering radius is determined by the Voronoi vertex that lies furthest from \set{P}.

%
%fig:spherical_covering_circles was here
\begin{figure}%[!h]
	\includegraphics[width=\linewidth]{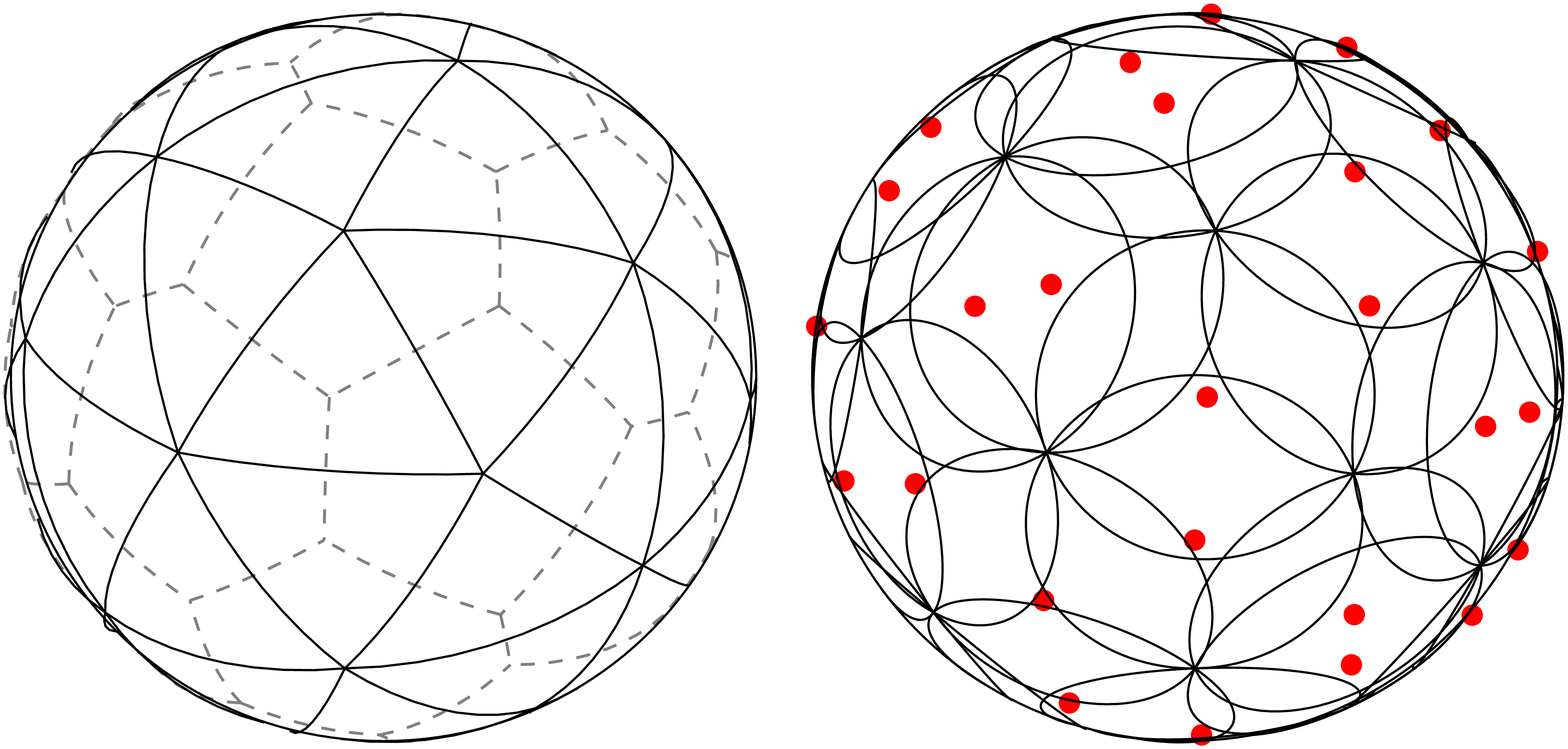}
	\hcaption{\textbf{Left:} a putatively optimal spherical covering for 28 points in \vectorspace{S}{2} (point set due to Hardin et al.~\shortcite{SloaneCoverings}).  The solid lines indicate the spherically constrained Delaunay triangulation.  The dashed lines indicate the Voronoi cells.\\
\textbf{Right:} the same points; each simplex in the Delaunay triangulation has a circumcap, the centre of which (marked in red) lies at a Voronoi cell vertex.  The maximum simplex circumradius determines the covering radius of the point set.}
\label{fig:spherical_covering_circles}
\end{figure}

The Voronoi cells of a set of points in \vectorspace{S}{d} are not easy to calculate directly, so instead we calculate the (hyperspherically constrained) Delaunay triangulation.  The Delaunay triangulation, $DT(\set{P})$, is a set of hyperspherical simplices whose vertices are points in \set{P} which satisfy the empty-sphere condition, that is, no points in \set{P} lie inside the circumhypercap of any simplex in $DT(\set{P})$.  Associated with each simplex is a Voronoi vertex, which lies at the centre of the simplex circumhypercap (the circumcentre).  The Delaunay triangulation, Voronoi cells and simplex circumhypercaps and circumcentres are illustrated in \vectorspace{S}{2} in Figure~\ref{fig:spherical_covering_circles}.  We now show how to calculate the circumcentre of a simplex.
\begin{theorem*}
For a hyperspherical simplex $\set{t} \in DT(\set{P})$ with vertices $\{ \point{1}, \point{2}, \point{3}, \ldots, \point{d+1} \} \in \vectorspace{S}{d}$, the position of the circumcentre, \circumcentre, is equal to the unit normal vector of the $d$-dimensional hyperplane on which the vertices of \set{t} lie.
\end{theorem*}
\begin{proof}
Let $\set{S} = \{ s_1 = \point{2} - \point{1}, s_2 = \point{3} - \point{1}, s_3 = \point{4} - \point{1}, \ldots, s_{d} = \point{d+1} - \point{1}\}$ and let $\circumcentre \in \vectorspace{S}{d}$ be the circumcentre of \set{t}.  Then, per definition, \circumcentre must satisfy:
\begin{equation}
\point{i} \cdot \circumcentre = \point{1} \cdot \circumcentre \;\;\;\; \forall i
\label{eq:circumcentre_definition}
\end{equation}
Subtracting $\point{1} \cdot \circumcentre$ from each side gives:
\begin{equation}
s_{i} \cdot \circumcentre = 0 \;\;\;\; \forall i
\label{eq:circumcentre_difdot}
\end{equation}
The unit length of $X$ follows from requiring $X \in \vectorspace{S}{d}$.
\end{proof}
To find \circumcentre, we calculate the normalized $d$-fold vector cross product~\cite{brown1967vector} of \set{S}.  Since every hyperplane has two (opposite) unit plane normals, $X$ has two solutions, which correspond to the centre of the simplex hypercircumcap and its dual.  However, given that $|\set{P}| \geq d+2$  only one of these solutions fulfils the empty-sphere condition, which is the one which satisfies: $\left\langle X, \point{i} \right\rangle > 0 \;\;\forall i$.  This corresponds to the smaller of the two hypercircumcaps.

%
%fig:spherical_convex_hull was here
\begin{figure}%[!h]
	\includegraphics[width=\linewidth]{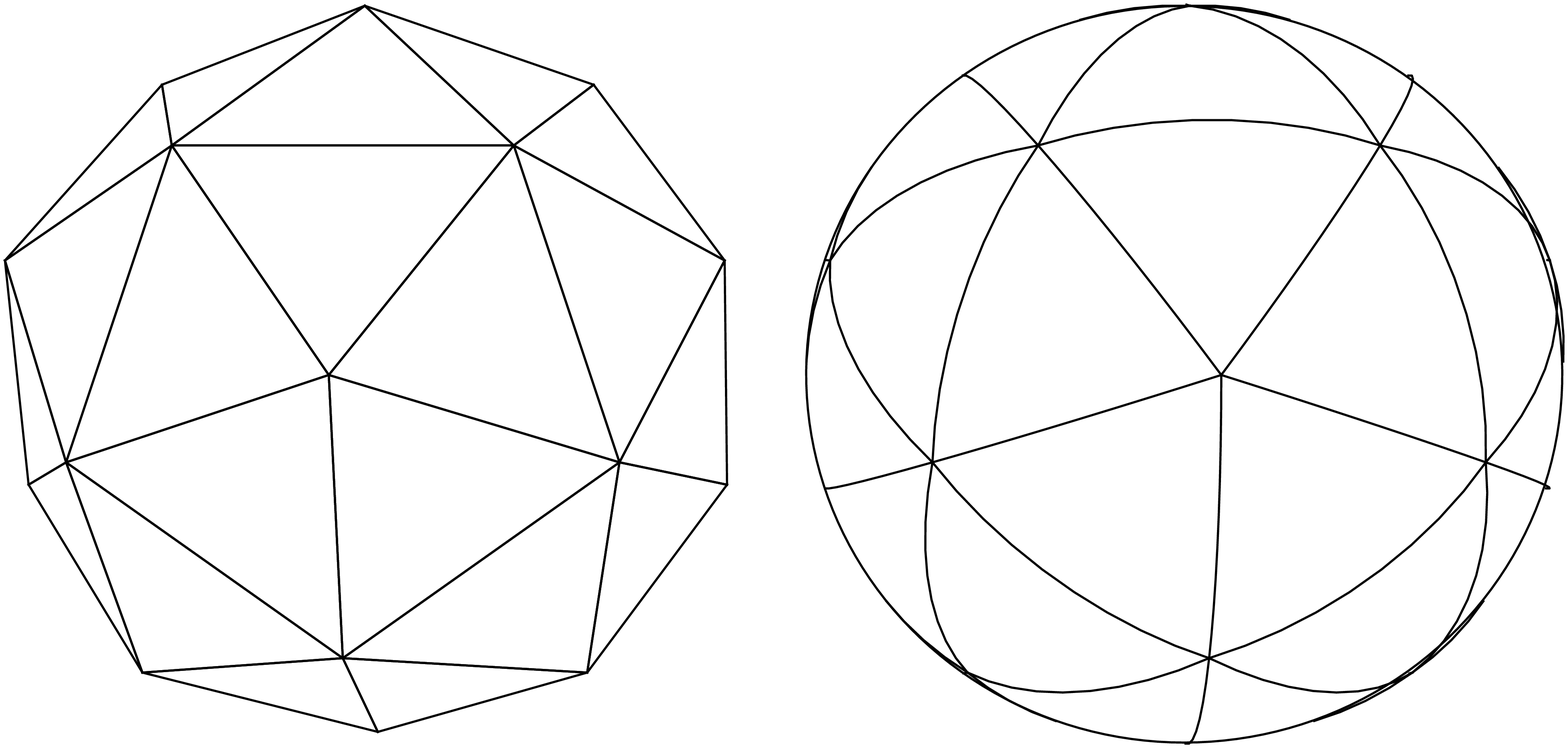}
	\hcaption{Convex hull (left) and the spherically constrained Delaunay triangulation (right) of 22 points on the sphere. The triangulations exist in \vectorspace{R}{3} and \vectorspace{S}{2} respectively, but the vertices of each simplex are the same.  Data due to Hardin et al.~\shortcite{SloaneCoverings}.}
\label{fig:spherical_convex_hull}
\end{figure}

For a set of points on \vectorspace{S}{d}, the vertices of each simplex $\set{t} \in DT(\set{P})$ can be found by calculating the convex hull of \set{P}, as shown in Figure~\ref{fig:spherical_convex_hull}.  If we denote the circumradius of a simplex \set{t} by $\phi(\set{t})$, Equation~(\ref{eq:covering_radius}) can be restated as:
\begin{equation}
\radius = \max\{\phi(\set{t}) \mid \set{t} \in DT(\set{P}) \}
\label{eq:covering_radius_measurement}
\end{equation}
which provides a practical solution to Equation~(\ref{eq:covering_radius}): the covering radius of a point set is simply the maximum simplex circumradius.

\subsection{Orientation Sets}
\label{sec:orientations_sets}

The problem of finding a good spherical covering is immediately relatable to the problem of finding good sets of orientations.  As described previously, rotations can be represented by quaternions, which are points on \vectorspace{S}{3}.
%This is possible since $\vectorspace{S}{3}$ is a double covering of \sothree~\cite{altmann2005rotations}.
The maximum rotational angle between a point $x \in \sothree$ and a point set $\set{P}$, also called the maximum \textit{misorientation}, is given by:
\begin{equation}
\begin{array}{l}
\alpha_{\text{max}} = 2 \max\limits_{x \in \vectorspace{S}{3}} \min\limits_{\point{} \in \set{P}} \min \left[ \arccos \langle -x, \point{} \rangle, \arccos \langle x, \point{} \rangle \right] \\
\phantom{\alpha_{\text{max}}} =2 \max\limits_{x \in \vectorspace{S}{3}} \min\limits_{\point{} \in \set{Q}}\arccos \langle x, \point{} \rangle
\end{array}
\label{eq:misorientation}
\end{equation}
where
$\set{Q} = \set{P} \cup \{ -\point{} \mid \point{} \in \set{P}\}$.
It can be seen that, for a point set with antipodal symmetry, $\alpha_{\text{max}} = 2 \theta$, that is, the maximum misorientation is twice the covering radius.  Thus, the problem of finding a set of rotations with the lowest maximum misorientation is equivalent to finding an optimal spherical covering for a point set with antipodal symmetry on \vectorspace{S}{3}.

\subsection{Integration with Crystallographic Symmetries}
\label{sec:theory_symmetry}

Equation~(\ref{eq:misorientation}) shows that a set of $2N$ points with antipodal symmetry represents a set of $N$ rotations.  A set of orientations generated in this way covers the whole space of \sothree, and is immediately applicable to pattern indexing of materials with triclinic ($C_1$) Bravais lattices.  For materials with higher order symmetry, though, a dictionary set which covers all of \sothree is wasteful, since only the fundamental zone orientations~\cite{he2007representation} are needed.  A naive approach for selecting fundamental zone orientations is to generate a full covering of \sothree and then simply `cut out' the desired region; this introduces artifacts at the boundaries of the fundamental zone which increase the covering radius significantly.  Instead, we apply the symmetry of the desired point group during generation of the orientation sets.

Given a set of basis points $\set{B} = \{b_{1}, b_{2}, \ldots \}$ and a quaternion group $\set{G} = \{g_{1}, g_{2}, \ldots \}$, we can create a set of points with the symmetry of \set{G} by:
\begin{equation}
\set{P} = \{ b \otimes g \mid b \in \set{B}, g \in \set{G} \}
\label{eq:symmetry}
\end{equation}
where $\otimes$ denotes quaternion multiplication.
If \set{P} is to represent a set of orientations (c.f.~Equation~(\ref{eq:misorientation})), \set{G} must be a superset of antipodal symmetry ($C_1$).  The finite quaternion groups which meet this requirement are~\cite{conway2003quaternions}:
%\begin{table}[h]
\begin{center}
\begin{tabular}{cl}
$2I_{60}$ & The binary icosahedral group\\
$2O_{24}$ & The binary octahedral group\\
$2T_{12}$ & The binary tetrahedral group\\
$2D_{n}$ & The binary dihedral group\\
$2C_{n}$ & The binary cyclic group\\
\end{tabular}
\end{center}
With the exception of the binary icosahedral group, each of these is used to describe the generators of the 11 Laue groups~\cite{morawiec2003orientations}, $C_1, C_2, C_3, C_4, C_6, D_2, D_3, D_4, D_6, T$ and $O$.
By the application of a symmetry group, the problem of finding a good spherical covering for a chosen crystallographic fundamental zone is reduced to a problem of finding an optimal configuration of the \emph{basis} points, which is a much smaller problem.

The Laue groups can be divided into two sets:
\begin{equation}
\{	C_2, C_4, D_2, D_4, T, O\}
\end{equation}
and
\begin{equation}
\{C_3, C_6, D_3, D_6\}
\end{equation}
where the elements of each are subsets of $O$ and $D_6$ respectively ($C_1$ is trivially a subset of both).  This means that, if we generate sphere coverings with $O$ and $D_6$ applied according to Equation~(\ref{eq:symmetry}), then by an appropriate mapping of the fundamental zone orientations we obtain sphere coverings for \emph{all} Laue groups, without the aforementioned boundary artifacts.  The Laue group subset relationships are shown in Appendix~\ref{sec:laue_tables}.

\subsection{Derivation of the Simplex Bound on \vectorspace{S}{3}}
\label{sec:lowerbound}

In addition to knowing the covering radius and density of a point set, it is useful to know how far from optimality a set is.  We can estimate the optimality gap with a lower bound.

The simplex bound is a classic result which gives an upper bound on the density of sphere packings, and a lower bound on the density of sphere coverings.  It has been proven for packings in \vectorspace{R}{d}~\cite{rogers1958packing} and \vectorspace{S}{d}~\cite{boroczky1978packing}, and for coverings in \vectorspace{R}{d}~\cite{coxeter1959covering} and \vectorspace{S}{2}~\cite{toth1964regular}. %toth1953lagerungen
B\"or\"oczky has conjectured that it is a lower bound on \vectorspace{S}{3}~\cite{boroczky2004finite}.  Despite lacking a proof, we will use the simplex bound on \vectorspace{S}{3} to estimate the optimality of our point sets, as it is `intuitively obvious'.

%
%fig:simplex_bound_r2 was here
\begin{figure}%[!h]
    \centering
    	\includegraphics[width=\linewidth]{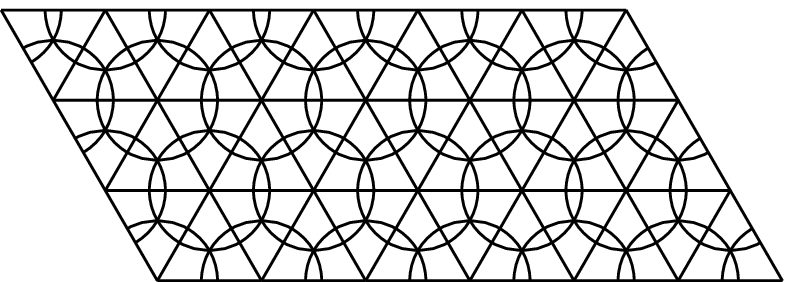}
	\hcaption{Illustration of the simplex bound in \vectorspace{R}{2}, shown here due to the difficulty of visualizing the simplex bound in \vectorspace{S}{3}.  Regular simplices in \vectorspace{R}{2} are equilateral triangles, which tessellate.  At the vertices of each triangle (of circumradius $r$) is a circle of radius $r$.  The area of intersection between a triangle and a circle is a circular sector of angle $\frac{\pi}{3}$.  Each triangle is covered by three equal areas of intersection.  The covering density is therefore the ratio of the sum of the three areas of intersection to the area of the triangle: $\tau_{\vectorspace{R}{2}} = \frac{2\pi}{3\sqrt{3}}$.  In \vectorspace{R}{d} the covering density is independent of $r$, which is not the case in \vectorspace{S}{d} for $d \geq 2$ due to a lack of tesselation.}
\label{fig:simplex_bound_r2}
\end{figure}

The premise of the simplex covering bound is that the lowest covering density can be achieved with regular simplices; this concept is illustrated in Figure~\ref{fig:simplex_bound_r2}.  Regular simplices tesselate
%(they fill space with no overlap)
in \vectorspace{R}{1} and \vectorspace{R}{2}.  In \vectorspace{R}{d} for $d \geq 3$, regular simplices do not tesselate, and thus the simplex covering density is an unattainable lower bound.  As stated previously, regular simplices tesselate in \vectorspace{S}{2} for three configurations: the tetrahedron, the octahedron and the icosahedron.  Thus, the simplex bound is tight for these configurations only, and is provably unattainable for any other number of vertices.  In \vectorspace{S}{3}, regular simplices tesselate only in the 5-cell, the 16-cell and the 600-cell.  If B\"or\"oczky's conjecture is correct, the simplex bound is tight only for these configurations. Since no description of the simplex bound covering density on \vectorspace{S}{3} could be found in the literature, we derive an expression for it here.

Given a hyperspherical cap on \vectorspace{S}{3} of radius \radius and volume $C_{3}(\radius)$, we denote the inscribed regular spherical tetrahedron $T(\radius)$.  At each of the four vertices of $T(\radius)$ is a hyperspherical cap of radius \radius.  Each of these caps intersects $T(\radius)$ with solid angle $\Omega(\radius)$, giving a volume of intersection of $C_{3}(\radius)\frac{\Omega(\radius)}{4 \pi}$. Now $T(\radius)$ is covered by the four equal volumes of intersection.
The covering density, $\tau_{\vectorspace{S}{3}}$, is the ratio of the sum of the four volumes of intersection to the volume of $T(\radius)$:
\begin{equation}
\tau_{\vectorspace{S}{3}}(\radius) = 4 \; C_{3}(\radius) \; \frac{\Omega(\radius)}{4 \pi} \frac{1}{\text{Vol}(T(\radius))}
\end{equation}
where:
\begin{equation}
C_{3}(\radius) = \pi (2\radius - \sin(2 \radius)) \label{eq:density_bvol}
\end{equation}
\begin{equation}
\Omega(\radius) = 3\psi(\radius) - \pi \label{eq:density_solidangle}
\end{equation}
\begin{equation}
\psi(\radius) = \arccos\left( \frac{4\cos^2(\radius) - 1}{8\cos^2(\radius) + 1} \right) \label{eq:density_dihedral}
\end{equation}
\begin{multline}
\text{Vol}(T(\radius)) = \biggl( -\operatorname{Re}(L) + \pi (\arg(-Q)\\+ 3 \psi(\radius)) - \frac{3}{2}\pi^2 \biggr) \mod 2 \pi^2 \label{eq:density_tet0}
\end{multline}
\begin{equation}
Q = 3e^{-2i\psi(\radius)} + 4e^{-3i\psi(\radius)} + e^{-6i\psi(\radius)} \label{eq:density_tet1}
\end{equation}
\begin{multline}
L = \frac{1}{2} \biggl[ \dilog{Z_0} + 3\dilog{Z_0 e^{-4i\psi(\radius)}}\\ - 4\dilog{-Z_0 e^{-3i\psi(\radius)}} - 3\psi(\radius)^2 \biggr] \label{eq:density_tet2}
\end{multline}
\begin{multline}
Z_0 = \frac{-6\sin^2(\psi(\radius))}{Q}\\ + \frac{2\sqrt{(\cos(\psi(\radius))+1)^3 (1 - 3\cos(\psi(\radius)))}}{Q} \label{eq:density_tet3}
\end{multline}
where $\psi(\radius)$ is the dihedral angle of $T(\radius)$.
The terms in Equations~(\ref{eq:density_bvol}) - (\ref{eq:density_dihedral}) are derived in \ref{app:simplex_bound_derivation}.  Equations (\ref{eq:density_tet0}) - (\ref{eq:density_tet3}) are a simplification of Murakami's formula for the volume of a spherical tetrahedron~\cite{murakami2012volume}, for the case where all six dihedral angles are equal (a regular spherical tetrahedron).

The covering density can be used to estimate the optimality gap of a point set.  For a set of $N$ points with covering radius \radius, the lower bound on the covering radius $\radius^*$ can be found by rearranging the density expression in Equation~(\ref{eq:spherical_density}):
\begin{equation}
N = \frac{2\pi^2 \tau_{\vectorspace{S}{3}}(\radius^*)}{C_{3}(\radius^*)}
\end{equation}
where $2\pi^2$ is the surface area of \vectorspace{S}{3}.  The optimality gap of the point set is then $\radius / \radius^* - 1$.  Since $\tau_{\vectorspace{S}{3}}(\radius^*)$ is a nontrivial expression, we find $\radius^*$ numerically.

\section{Method of Orientation Set Generation}
\label{sec:method}

We now describe the method for generating point sets with small covering radii.  The direct problem formulation with the application of symmetry is shown in Table~\ref{model:direct_model}.  This is essentially just a restatement of Equations~(\ref{eq:covering_radius_measurement}) and (\ref{eq:symmetry}).

%
%model:direct_model was here
\begin{table}
\begin{tabularx}{\textwidth}{ll}
\hline
\\
\textbf{Variables:}& $\set{B} = \{b_{1} \in \vectorspace{S}{3}, b_{2} \in \vectorspace{S}{3} \ldots \}$\Bstrut{3}\\% &Points in basis set
\textbf{Parameters:}	& $\set{G} = \{g_{1} \in \vectorspace{S}{3}, g_{2} \in \vectorspace{S}{3}, \ldots \}$\Bstrut{3}\\% &Chosen symmetry group
\textbf{Minimize:}& $\radius = \max\{\phi(\set{t}) \mid \set{t} \in DT(\set{P}) \}$\Bstrut{3}\\%	&Minimize covering radius
\textbf{Subject to:}	& $\set{P} = \{ b \otimes g \mid b \in \set{B}, g \in \set{G} \}$\Bstrut{2}\\% &Point set is composition of \set{B} and \set{G}
\hline
\end{tabularx}
\hcaption{Direct model for minimizing the covering radius of a point set in \vectorspace{S}{3}.  The point set \set{P} is composed of a basis set, \set{B}, on which a chosen symmetry group, \set{G}, acts.  The covering radius, $\radius$, is calculated using the Delaunay triangulation of \set{P}.\label{model:direct_model}}
\end{table}

The problem of finding optimal spherical coverings is difficult; in addition to being a NP-hard problem~\cite{van1981another}, the objective function is non-differentiable, and the `fitness landscape' is non-convex and has many local minima.  One possible solution approach (used by Hardin et al.~\shortcite{SloaneCoverings} to generate coverings in \vectorspace{S}{2}) is to use direct search.  This overcomes the non-differentiability of the objective function, but repeated solution from many different starting configurations is required
to find the globally optimal configuration.
%In addition, direct search methods scale very poorly with increasing problem size; depending on the vertex representation used, at least $3N$ variables are needed to represent a point set in \vectorspace{S}{3} and direct search becomes infeasible for large $N$.  The application of a symmetry group reduces the number of variables in the basis set, though only by a constant factor.
Furthermore, due to the poor scaling of direct search methods with increasing problem size, this approach is not practical since we wish to create very large orientation sets.

Since it is unlikely that we will find globally optimal solutions for large point sets with direct search, we will instead attempt to find good solutions with an indirect method.  We proceed as follows: an initial set of orientations is created by sampling randomly from a uniform distribution on \sothree~\cite{Shoemake1992}.  The covering radius is then succesively reduced, firstly by using gradient descent to find a configuration which is a local minimizer of the Riesz energy.  Secondly, a smoothing procedure is used to improve the characteristics of the Delaunay triangulation.  Lastly, a local optimization procedure is used to further refine the solution.  We present no theoretical basis for the choice of methods, nor for the order in which the methods are applied.  Rather, empirical experimentation has shown that the method is effective and produces point sets with a small covering radius.

The motivation for choosing these methods is illustrated in Figure~\ref{fig:process_illustration}.
The random point set has a large covering radius.  By minimizing the Riesz energy the covering radius is significantly reduced.  The covering radius can be further reduced as shown in the optimal covering.  The effect of the smoothing procedure is not shown here, as it is visually very similar to the Riesz energy and optimal covering configurations.
In the rest of this Section we describe each method in detail.

%fig:process_illustration was here
\begin{figure}%[!h]
\centering
\includegraphics[width=\textwidth]{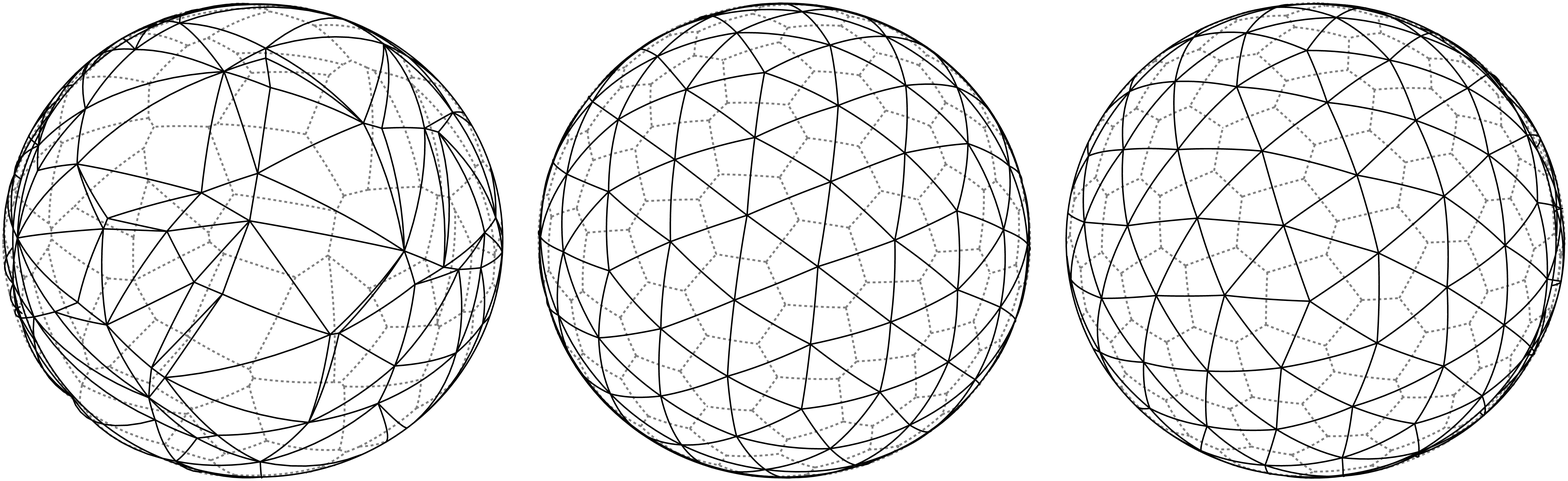}
\hcaption{Delaunay triangulations and Voronoi cells of three different point sets in \vectorspace{S}{2}, for $N=130$. \textbf{Left:} Points sampled uniformly from a random distribution. \textbf{Centre:} The global minimum configuration for the Riesz energy, here for $s=1$ (point set due to Wales \& Ulker~\shortcite{wales2006structure}).  \textbf{Right:} Putatively optimal spherical-covering configuration (point set due to Hardin et al.~\shortcite{SloaneCoverings}).  Point sets in \vectorspace{S}{2} are used here for illustrative purposes only, due to the difficulty of visualizing \vectorspace{S}{3}.}
\label{fig:process_illustration}
\end{figure}

\subsection{Riesz Energy Minimization}

For a set of points $\pointset \in \vectorspace{S}{d}$, the Riesz energy is defined as:
\begin{equation}
E_s(\mathbf{P}) =
\left\{ \begin{array}{ll}
\sum\limits_{i \neq j}^{N} \frac{1}{|\point{i} - \point{j}|^s} & \text{if}\ s > 0 \\
\sum\limits_{i \neq j}^{N} \log \frac{1}{|\point{i} - \point{j}|} & \text{if}\ s = 0 \\
\end{array}\right.
\end{equation}

The problem of finding optimal Riesz energy configurations is well studied, most commonly for $(d=3, s=1)$ (also known as the Thomson problem)~\cite{erber1991equilibrium, altschuler1994method, wales2006structure}, but also for $(d=4, s=1)$~\cite{altschuler2007symmetric}, and in the general case~\cite{hardin2004discretizing, rakhmanov1994electrons}.  The sphere-packing problem is equivalent to solving for $s=\infty$.
%todo: s = d-1 is called the newtonian potential

Cohn \& Kumar~\shortcite{Cohn2007} have shown that there exist configurations for certain values of $N$ which are universally optimal, that is, globally optimal solutions for every value of $s$.  The known universally optimal configurations for $d=3$ are the tetrahedron, the 16-cell and the 600-cell. The vertices of these polyhedra are conjectured to be global optima for the sphere-covering problem, since their Delaunay triangulations consist of regular spherical tetrahedra (c.f. Section~\ref{sec:lowerbound}).  However, for any value of $N$ for which a universally optimal configuration does not exist, there is no value of $s$ which for a configuration minimizing $E_s(\mathbf{P})$ guarantees an optimal spherical covering.  As such, we will select a value of $s$ on the following basis:  Kuijlaars et al.~\shortcite{kuijlaars2007separation} have shown that the set of points \set{P} which minimizes $E_s(\mathbf{P})$ is well-distributed when $d - 1 \leq s < d$.  We will select $s=2$ since longer range potentials exhibit fewer local minima~\cite{wales2006structure}.  We have used the PR+ conjugate gradient method~\cite{nocedal1999numerical} to find a local minimum of $E_s(\mathbf{P})$.  The resulting configuration is a good intermediate solution with a small covering radius.

\subsection{Optimal Delaunay Triangulation Smoothing}

Minimizing the Riesz energy of a point set reduces the covering radius whilst considering only the relative positions of the points.  We can obtain a further reduction in covering radius by considering the positions of a point set \textit{and} the simplices in its Delaunay triangulation.  This is a well-studied problem in the computational geometry community known as \emph{tetrahedral meshing}.  Given a set of points sampled from an object (e.g. a teapot model) the objective is to move the points in order to create a `nice' Delaunay triangulation (the mesh) whilst preserving the shape of the object.  Chen~\shortcite{chen2004optimaldt} defines an optimal Delaunay triangulation as a set of points which minimizes the energy function:
\begin{equation}
E_{\text{ODT}} = \frac{1}{d+1} \sum\limits_{i=1...N}\int_{\Omega_i} || p - p_i||^2 dp
\end{equation}
where $\Omega_i$ is the 1-ring of $p_i$ (the volume bounded by $p_i$ and its simplicial neighbours).  Minimization of this energy results in a Delaunay triangulation whose simplices have a low circumradius to inradius ratio.  Alliez et al.~\shortcite{alliez2005variational} have shown that, for a given point, the position which minimizes $E_{\text{ODT}}$ is:
\begin{equation}
p_{i}^{*} = \frac{1}{\text{Vol}\left( \Omega_i \right)} \sum\limits_{\set{t} \in \Omega_i}\text{Vol}(\set{t})C(\set{t})
\label{eq:smoothing_optimal_position_rd}
\end{equation}
where $\text{Vol}(\set{t})$ and $C(\set{t})$ are respectively the volume and circumcentre of simplex \set{t}.
They have shown that the energy can be minimized with guaranteed convergence by alternately constructing the Delaunay triangulation, and moving the vertices to their optimal positions using Equation~(\ref{eq:smoothing_optimal_position_rd}).

For our applications the `object' whose shape we must preserve is simply \vectorspace{S}{3}.  As such, after calculating the optimal vertex position using Equation~(\ref{eq:smoothing_optimal_position_rd}) the vertex position is normalized in order to bring it back onto \vectorspace{S}{3}.  We also calculate $\text{Vol}(\set{t})$ for a spherical tetrahedron~\cite{murakami2012volume} rather than for a Euclidean tetrahedron.  Despite the intended use for Euclidean geometries, we have found that this method works very well in practice in \vectorspace{S}{3}, which is likely due to the small local curvature of \vectorspace{S}{3} for large point sets.

\subsection{Local Refinement}

As a last step in the process of reducing the covering radius, we use an optimization procedure to iteratively refine a succession of local neighbourhoods.  We do so by generalizing the direct problem, by iteratively dividing \set{B} into an active set \set{A} and a constant set \set{C}.  We then minimize the maximum circumradius of the simplices with a vertex in \set{A}.  A description of the optimization problem is given in Table~\ref{model:local_refinement_model}.

%
%model:local_refinement_model was here
\begin{table}
\begin{tabularx}{\textwidth}{llr}
\hline
\\
\textbf{Variables:}
& $\set{A} = \{a_{1} \in \vectorspace{S}{3}, a_{2} \in \vectorspace{S}{3}, \ldots \}$ & (1)\Bstrut{3}\\
\textbf{Parameters:}
& $\set{C} = \{c_{1} \in \vectorspace{S}{3}, c_{2} \in \vectorspace{S}{3}, \ldots \}$ & (2)\Bstrut{2}\\
& $\set{G} = \{g_{1} \in \vectorspace{S}{3}, g_{2} \in \vectorspace{S}{3}, \ldots \}$& (3)\Bstrut{3}\\
\textbf{Minimize:}
& $\max\{\phi(\set{t}) \mid \set{t} \in DT(\set{P}) \; \land \; \set{t} \cap \set{A} \neq \emptyset\}$&(4)\Bstrut{3}\\
\textbf{Subject to:}
& $\set{B} = \set{A} \cup \set{C}$&(5)\Bstrut{2}\\
& $\set{P} = \{ b \otimes g \mid b \in \set{B}, g \in \set{G} \}$&(6)\Bstrut{2}\\
\hline
\end{tabularx}
\hcaption{Model for reducing the covering radius of a \emph{local} neighbourhood of a point set.  The point set \set{P} is composed of a basis set, \set{B}, on which a chosen symmetry group, \set{G}, acts.  The basis set, \set{B}, consists of an active set, \set{A}, which defines the local neighbourhood to be optimized, and a constant set, \set{C}, which contains the remaining points.  The covering radius, $\radius$, is again calculated using the Delaunay triangulation, though only of the points which are either active or which share a simplicial neighbour with an active point.}
\label{model:local_refinement_model}
\end{table}

Whilst the smallest active set consists of a single vertex, we find that optimizing the vertices of a whole simplex at a time gives better results.  To do so, we alternately construct the Delaunay triangulation, and then optimize each simplex in turn.  The order in which the simplices are optimized is determined by their circumradius, from largest to smallest.  After each update the chosen symmetry group is reapplied to the basis set in order to maintain a consistent point set.

Since the minimization the maximum value of a set is a non-differentiable objective function, we use the Nelder-Mead method~\cite{nelder1965simplex} to optimize the above function as it is a derivative-free method.  In order to avoid dealing with the implicit constraint $|p| = 1 \;\; \forall p \in \set{P}$, we represent the vertices using RF vectors.  Representing the vertices as RF vectors during optimization has the added benefit of reducing the number of variables, which is particularly advantageous when using the Nelder-Mead method.  Since a RF vector representation of any $180^\circ$ rotation has infinite magnitude, we rotate the local neighbourhood under consideration to $\{1, 0, 0, 0\}$ prior to optimization, and back again after optimization.

\section{Results}
\label{sec:results}

Figure~\ref{fig:hist_comparison} illustrates how each stage of the optimization process affects the solution quality.  The initial random sampling results in a distribution of simplex circumradii that is approximately Gaussian.  Minimization of the Riesz energy significantly reduces the mean and variance of the simplex circumradii, as well as the number of simplices.  The distribution resembles a bimodal Gaussian distribution, which suggests an ordered underlying simplex structure.
%the two components of the bimodal distribution can in fact be separated by using a variance-minimizing segmentation of the ratios of the maximum vertex angle to the minimum vertex angle of each simplex.  
Application of ODT smoothing reduces the mean and variance of the of simplex circumradii, and results, again, in an approximately Gaussian distribution.  Lastly, the objective of the local refinement procedure is to minimize the maximum simplex circumradius.  It can be seen that this produces a peak around the maximum circumradius with a tail of smaller circumradii below this.

%fig:hist_comparison was here
\begin{figure}%[!h]
\includegraphics[width=\textwidth]{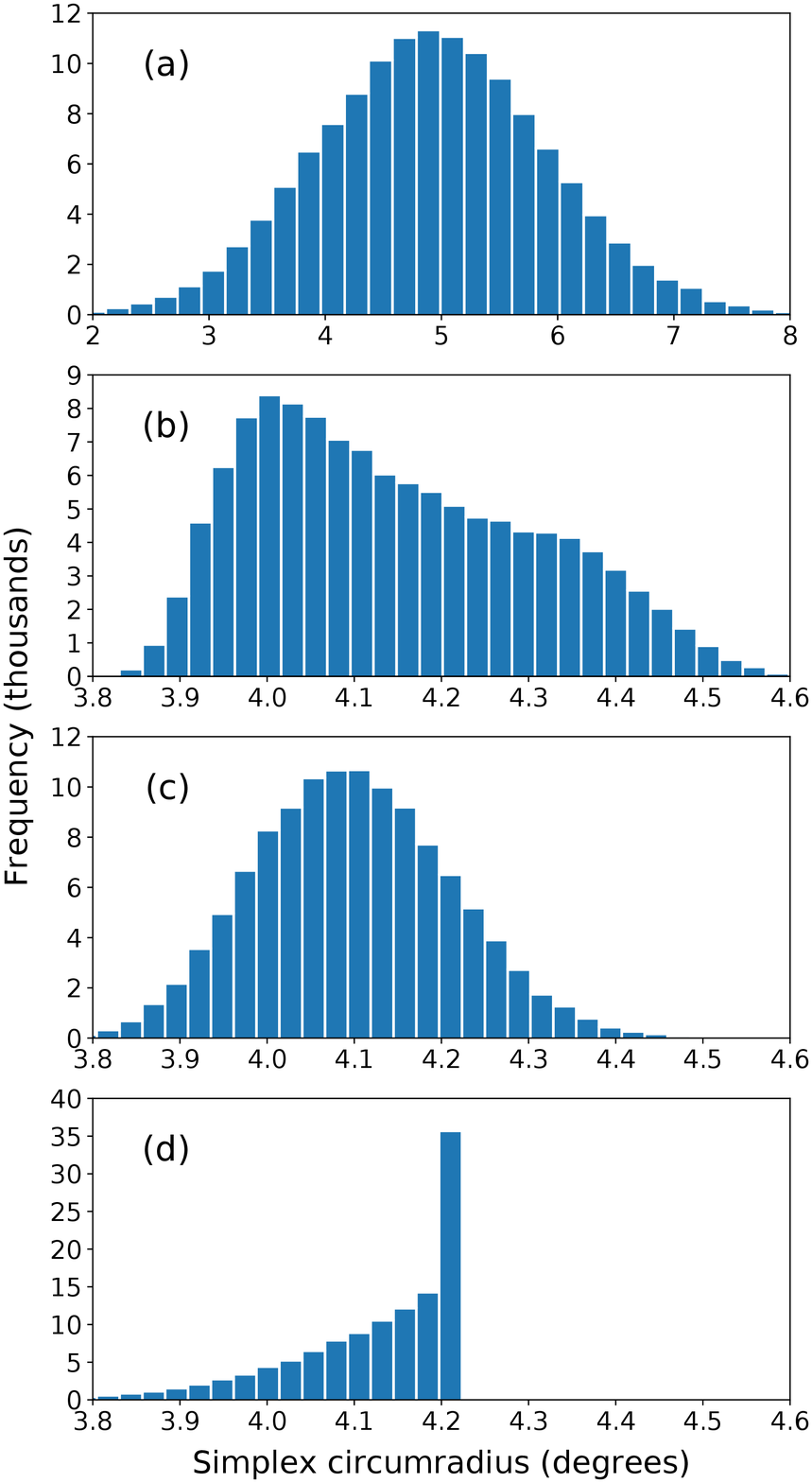}
\hcaption{Histograms showing the change in simplex circumradius at each stage in the optimization for a point set with antipodal symmetry and $N=20000$.  The histograms show the simplex circumradius distribution after: (a) initial random sampling, (b) Riesz energy minimization, (c) ODT smoothing, (d) local refinement.  The maximum circumradius is reduced at every stage.}
\label{fig:hist_comparison}
\end{figure}

Figure~\ref{fig:comparison_all} shows a comparison of our method with the methods discussed in Section~\ref{sec:introduction}, in the range $N=[960, 200000]$.
We have applied $2I_{60}$ symmetry, as it requires a small basis set and thus allows us to quickly generate  coverings of the full space of \vectorspace{S}{3}.
For each value of $N$, we have applied our method from 200 random starting configuration and taken the point set with the lowest covering radius.  It can be seen that the resulting sets have a lower covering radius than the other methods, both at small and large values of $N$.
%It typically achieves the same covering radius with $\approx 80\%$ of the number of points of BCC grids with binary octahedral symmetry, $\approx 30\%$ of incremental grids based on the Hopf fibration, and $\approx 10\%$ of random sampling from a uniform distribution.
Furthermore, our method displays a smooth decrease in covering radius with increasing $N$, which is highlighted by the almost constant covering density.  We do not claim optimality for any of our point sets; in most cases the covering radius of best point set was unique amongst the 200 runs.  As such we can conclude that lower covering radii could be obtained simply by increasing the number of runs, though this is very time consuming for large point sets.

%
%fig:comparison_all was here
\begin{figure}%[!h]
\includegraphics[width=\textwidth]{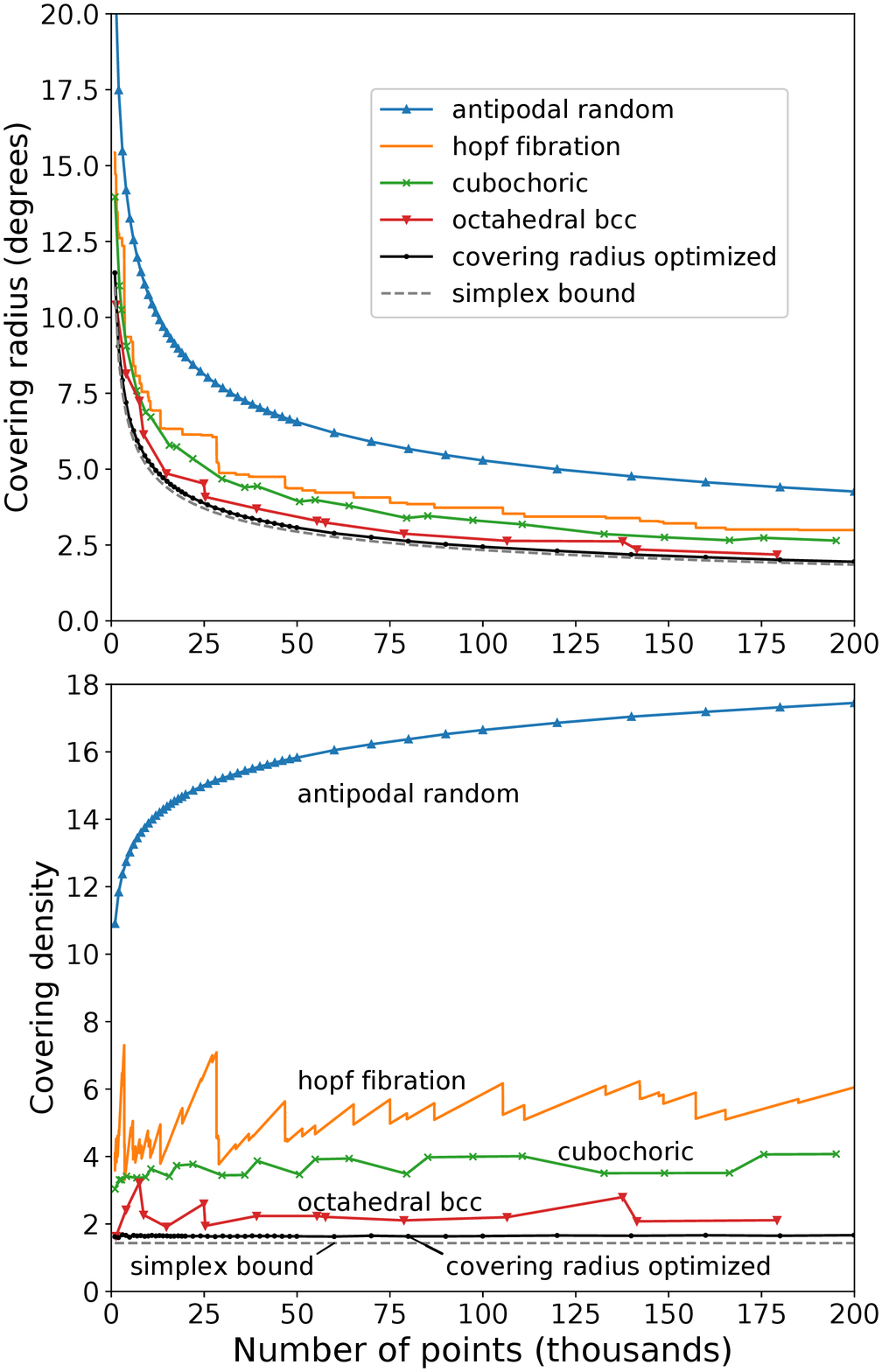}
\hcaption{Comparison of the covering radius (left) and the covering density (right) of random sampling from a uniform distribution with antipodal symmetry, incremental grids based on the Hopf fibration~\cite{Yershova2010}, cubochoric grids~\cite{rosca2014new}, BCC grids with binary octahedral symmetry~\cite{Karney2007}, and our method.  For the random sampling, the mean of $10^5$ runs was used.  For the incremental grids based on the Hopf fibration, the covering radius was calculated at every value of $N$ in the range shown.  For the covering radius optimized point sets (our method), the best result of 200 runs was used.}
\label{fig:comparison_all}
\end{figure}

The optimality gaps of some selected point sets generated using our method are shown in Table~\ref{table:optimality_gaps}.  The gaps are below $6\%$ at every value of $N$.
In the Euclidean limit ($N \rightarrow \infty$) the curvature of \vectorspace{S}{3} in a local area is effectively zero.  For this reason, the optimal covering in a local area should be a BCC lattice, since this is the best known covering in \vectorspace{R}{3}.  Since a BCC lattice has a higher covering density than the simplex bound, the optimality gaps presented here leave room for improvement.

%table:optimality_gaps was here
\begin{table}%[h]
\hcaption{Conjectured optimality gaps for covering radius optimized configurations, with $2I_{60}$ symmetry applied.  $N$ is the number of points in each set, $\radius$ is the covering radius, $\radius^*$ is the covering radius of the simplex bound on \vectorspace{S}{3}, conjectured to be a lower bound~\cite{boroczky2004finite}.  The optimality gap percentage is $100 \left( \radius / \radius^* - 1 \right)$. $^{\dagger}N=8$ and $N=120$ are the point sets containing the vertices of the 16-cell and 600-cell respectively, included here to highlight the tightness of the simplex bound for point sets consisting of regular tetrahedral cells.}
%$^\star$Large sets created with the aim of having $2\radius$ less than $1^{\circ}$, $2^{\circ}$ and $3^{\circ}$.  Only a single run was used to create these point sets.
\label{table:optimality_gaps}
\begin{tabular}{lrrr}
%\hline
$N$ & $\radius$ & $\radius^*$ & Opt. Gap\\
\hline
$8^{\dagger}$ & $60.00^{\circ}$ & $60.00^{\circ}$ & $0.00\%$ \\
$120^{\dagger}$ & $22.24^{\circ}$ & $22.24^{\circ}$ & $0.00\%$ \\%22.238756093
%960 & $11.47^{\circ}$ & $11.01^{\circ}$ & 4.15\% \\
1920 & $9.05^{\circ}$ & $8.73^{\circ}$ & 3.68\% \\
%3000 & $7.93^{\circ}$ & $7.52^{\circ}$ & 5.46\% \\
3960 & $7.20^{\circ}$ & $6.85^{\circ}$ & 5.05\% \\
%4920 & $6.62^{\circ}$ & $6.37^{\circ}$ & 3.93\% \\
6000 & $6.27^{\circ}$ & $5.96^{\circ}$ & 5.07\% \\
%6960 & $5.94^{\circ}$ & $5.68^{\circ}$ & 4.72\% \\
7920 & $5.71^{\circ}$ & $5.44^{\circ}$ & 4.95\% \\
%9000 & $5.44^{\circ}$ & $5.21^{\circ}$ & 4.50\% \\
9960 & $5.27^{\circ}$ & $5.04^{\circ}$ & 4.67\% \\
%10920 & $5.13^{\circ}$ & $4.88^{\circ}$ & 5.05\% \\
12000 & $4.96^{\circ}$ & $4.73^{\circ}$ & 4.71\% \\
%12960 & $4.84^{\circ}$ & $4.61^{\circ}$ & 4.93\% \\
13920 & $4.72^{\circ}$ & $4.50^{\circ}$ & 4.76\% \\
%15000 & $4.60^{\circ}$ & $4.39^{\circ}$ & 4.75\% \\
15960 & $4.50^{\circ}$ & $4.30^{\circ}$ & 4.54\% \\
%16920 & $4.41^{\circ}$ & $4.22^{\circ}$ & 4.60\% \\
18000 & $4.33^{\circ}$ & $4.13^{\circ}$ & 4.74\% \\
%18960 & $4.25^{\circ}$ & $4.06^{\circ}$ & 4.62\% \\
19920 & $4.18^{\circ}$ & $4.00^{\circ}$ & 4.61\% \\
%21960 & $4.04^{\circ}$ & $3.87^{\circ}$ & 4.56\% \\
24000 & $3.93^{\circ}$ & $3.76^{\circ}$ & 4.72\% \\
%25920 & $3.83^{\circ}$ & $3.66^{\circ}$ & 4.51\% \\
27960 & $3.72^{\circ}$ & $3.57^{\circ}$ & 4.31\% \\
%30000 & $3.65^{\circ}$ & $3.49^{\circ}$ & 4.67\% \\
31920 & $3.56^{\circ}$ & $3.41^{\circ}$ & 4.38\% \\
%33960 & $3.50^{\circ}$ & $3.34^{\circ}$ & 4.55\% \\
36000 & $3.43^{\circ}$ & $3.28^{\circ}$ & 4.47\% \\
%37920 & $3.38^{\circ}$ & $3.22^{\circ}$ & 4.78\% \\
39960 & $3.31^{\circ}$ & $3.17^{\circ}$ & 4.62\% \\
%42000 & $3.26^{\circ}$ & $3.12^{\circ}$ & 4.69\% \\
43920 & $3.21^{\circ}$ & $3.07^{\circ}$ & 4.67\% \\
%45960 & $3.16^{\circ}$ & $3.02^{\circ}$ & 4.52\% \\
48000 & $3.11^{\circ}$ & $2.98^{\circ}$ & 4.49\% \\
%49920 & $3.07^{\circ}$ & $2.94^{\circ}$ & 4.48\% \\
60000 & $2.89^{\circ}$ & $2.77^{\circ}$ & 4.35\% \\
%69960 & $2.76^{\circ}$ & $2.63^{\circ}$ & 4.84\% \\
79920 & $2.63^{\circ}$ & $2.51^{\circ}$ & 4.48\% \\
%90000 & $2.52^{\circ}$ & $2.42^{\circ}$ & 4.47\% \\
99960 & $2.44^{\circ}$ & $2.33^{\circ}$ & 4.64\% \\
%120000 & $2.31^{\circ}$ & $2.20^{\circ}$ & 5.02\% \\
139920 & $2.19^{\circ}$ & $2.09^{\circ}$ & 4.83\% \\
%159960 & $2.10^{\circ}$ & $1.99^{\circ}$ & 5.15\% \\
180000 & $2.01^{\circ}$ & $1.92^{\circ}$ & 4.84\% \\
%199920 & $1.95^{\circ}$ & $1.85^{\circ}$ & 5.21\% \\
%
%$^\star450000$ & $1.49^{\circ}$ & $1.41^{\circ}$ & $5.44\%$ \\
%$^\star1500000$ & $1.00^{\circ}$ & $0.95^{\circ}$ & $5.60\%$ \\
%$^\star12000000$ & $0.50^{\circ}$ & $0.47^{\circ}$ & $5.71\%$ \\
%\hline
\end{tabular}
\end{table}

\subsection{Practical Application}
The results presented in Figure~\ref{fig:comparison_all} demonstrate the evolution of the different methods with increasing size, though all at small sizes.  For a practical pattern-indexing application, much larger point sets are needed.  Furthermore, whilst the covering radius of a set specifies the maximum error, the distribution of errors is also of practical interest.  Figure~\ref{fig:hist_large} compares the error histograms of a covering radius optimized set and a cubochoric set, which is used for comparison due to its use in the widely used EMsoft microscopy software~\cite{EMsoft}.  In order to generate the error histogram $10^8$ random orientations were sampled; for each sampled orientation, the misorientation is calculated to the nearest orientation in the dictionary set.  A KD-tree~\cite{bentley1975multidimensional} is used to quickly find the closest dictionary orientation.  In addition to a smaller maximum error, the covering-radius optimized set has a better overall error distribution.  This is achieved despite the use of a smaller number of orientations.

%fig:hist_large was here
\begin{figure}%[!h]
\includegraphics[width=0.85\textwidth]{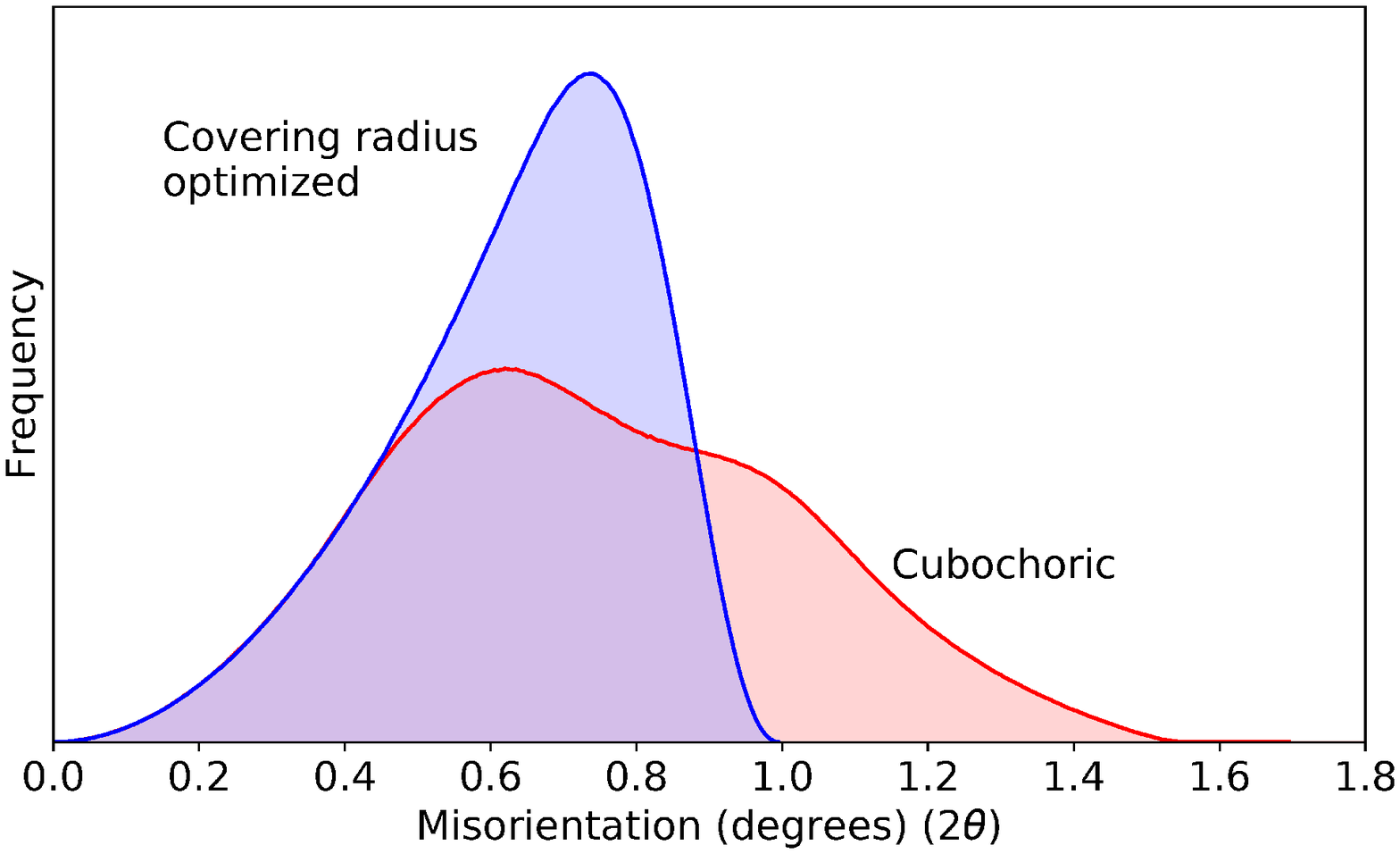}
\hcaption{Error histogram for a covering-radius optimized point set with $6 \times 10^6$ orientations and a cubochoric set with $6.3 \times 10^6$ orientations.  Here, both point sets cover the full space of \sothree, which corresponds to indexing a material with a triclinic crystal lattice.  The covering-radius optimized point set has a lower maximum error ($1.00^\circ$ vs. $1.72^\circ$) and a better overall distribution of errors.}
%cubochoric: 6028568 orientations 1.7159513 max error
\label{fig:hist_large}
\end{figure}

The maximum error of the covering radius optimized set is $72\%$ smaller than that of the cubochoric set.  In the Euclidean limit $\theta \propto n^{-1/3}$, which suggests that a cubochoric set would require approximately 5 times as many points to achieve the same maximum error.

Using the symmetry relationships described in Section~\ref{sec:theory_symmetry}, we have created orientation sets for every Laue group with maximum misorientations ($2\theta$) of $<1^\circ, 2^\circ, 3^\circ, 4^\circ$ and $5^\circ$, with optimality gaps less than $6\%$ for every set.  The orientation sets available online~\cite{githubrepository}.

\section{Summary}
\label{sec:summary}

We have shown how to construct a near-optimal sampling of orientations.  First we demonstrated that the sampling problem is equivalent to the problem of how to distribute points on a hypersphere.  We then showed that the best measure of quality for a point set is the covering radius, as this determines the maximum orientational error.  With the minimum covering radius as the objective, we created sets of orientations at a range of sizes for use in pattern indexing, and demonstrated that the number of orientation samples needed to achieve a desired indexing accuracy is significantly reduced as a consequence.

In addition to an exact calculation of the covering radius, which measures the quality of a set of orientations (smaller is better), we derived a lower bound on the covering radius, which sets a limit on the best attainable quality.  The difference between the achieved covering radius and the theoretical limit allows us to quantify the optimality of orientation sets, which we used to show that the sets we created are within $6\%$ of the optimal covering radius.

In order to use the method for indexing of diffraction patterns, we shown how symmetry groups can be imposed during orientation sampling, without introducing any edge-effect artifacts.  Using this approach we have demonstrated how to sample from the crystallographic fundamental zone of any of the 11 Laue groups.

Existing methods for sampling orientations have prioritized properties such as a refinable grid-like structure, fast generation, and the ability expand into spherical harmonics; we have instead chosen to optimize the maximum error (the covering radius) above all else.  This also means that the sampling method has very high computational requirements; the largest point set requires approximately 4 days of computation time.  Nonetheless, we claim that this is a good trade-off, since a point set must only be generated once for each desired error level, and affords a significant performance improvement every time a pattern is subsequently indexed.

     % Appendices appear after the main body of the text. They are prefixed by
     % a single \appendix declaration, and are then structured just like the
     % body text.

\appendix
\section{Simplex Bound Derivation}
\label{app:simplex_bound_derivation}
\subsection{Volume of a Hyperspherical Cap}

The volume of a hyperspherical cap in \vectorspace{S}{d} can be calculated by projection into RF space.  Since RF space is radially symmetric about the origin, the projection of a hyperspherical cap with radius \radius and centre coordinates $\{1, 0, 0, \ldots, 0 \} \in \vectorspace{S}{d}$ is a sphere with radius $r = \tan(\radius)$ centred at the origin.  Thus, the volume of the cap is the radial integral of the product of the surface area of a $(d-1)$-sphere with the RF space density:
$$C_{d}(\radius) = \int\limits_{0}^{\tan(\radius)} \frac{S_{d-1}(r)}{(1 + r^2)^2} dr
\hspace{2mm}
\text{where}
\hspace{2mm}
S_{d-1}(\radius) = \frac{d\pi^{d/2}}{\Gamma \left(\frac{d}{2} + 1 \right)}\radius^{d-1}
$$
For a hyperspherical cap in \vectorspace{S}{3}, this gives:
$$
C_{3}(\radius) = \int\limits_{0}^{\tan(\radius)} \frac{4\pi r^{2}}{(1 + r^{2})^2} dr = \pi (2\radius - \sin(2\radius))
$$
This is the same result derived by Moriawiec~\shortcite{morawiec2003orientations,morawiec2010volume}, but without normalization.

\subsection{Edge Length of a Regular Spherical Tetrahedron in \vectorspace{S}{3}}

Due to the radial symmetry of RF space, the RF projection of a regular spherical tetrahedron with centre coordinates $q_c = \{1, 0, 0, 0\} \in \vectorspace{S}{3}$ is a tetrahedron with centre coordinates $v_c = \{0, 0, 0\}$ and vertex coordinates:
$$
\begin{array}{ll}
v_1 = \{k,k,k\}
&v_2 = \{k,-k,-k\}\\
v_3 = \{-k,k,-k\}
&v_4 = \{-k,-k,k\}
\end{array}
$$
From this, we obtain the vertex coordinates in \vectorspace{S}{3}:
$$
\begin{array}{ll}
q_1 = \frac{1}{\sqrt{1 + 3k^2}} \{1,  k,  k,  k\}
&q_2 = \frac{1}{\sqrt{1 + 3k^2}} \{1,  k, -k, -k\}\\
q_3 = \frac{1}{\sqrt{1 + 3k^2}} \{1, -k,  k, -k\}
&q_4 = \frac{1}{\sqrt{1 + 3k^2}} \{1, -k, -k,  k\}
\end{array}
$$
The circumradius of the tetrahedron is given by the arc length from the centre to any of the vertices:
\begin{equation}
\radius = \arccos \langle q_c, q_i \rangle = \arccos \left( \frac{1}{\sqrt{1 + 3k^2}} \right) \;\;\;\; \forall i
\label{eq:quat_tet_radius}
\end{equation}
The edge length of the tetrahedron is the arc length between any two vertices:
\begin{equation}
l = \arccos \langle q_i, q_j \rangle = \arccos \left( \frac{1 - k^2}{1 + 3k^2} \right) \;\;\;\; \forall i \neq j
\label{eq:quat_tet_edge_length}
\end{equation}
%Solving for $k$ gives:
%%1 / (3*cost^2) - 1/3 = k^2
%$$\tan^2(\theta) = 3k^2$$
Using Equations~(\ref{eq:quat_tet_radius}) and (\ref{eq:quat_tet_edge_length}) we can express the edge length in terms of the radius:
\begin{equation}
l = \arccos\left( \frac{4 \cos^2(\radius) - 1}{3} \right)
\label{eq:radius_to_edge_length}
\end{equation}

\subsection{Dihedral Angle and Solid Angle of Intersection}

Let $\{q_1, q_2, q_3, q_4 \}$ be the vertices of a regular hyperspherical simplex in \vectorspace{S}{3} with the following coordinates:
$$
\begin{array}{ll}
q_1 = \left \{ 1,                          0, 0,  0  \right \}
&q_2 = \left \{ \cos l,  -a, \frac{-a}{\sqrt{3}}, z  \right \}\\
q_3 = \left \{ \cos l,   a, \frac{-a}{\sqrt{3}}, z  \right \}
&q_4 = \left \{ \cos l,   0, \frac{-2a}{\sqrt{3}}, z \right \}
\end{array}
$$
where:
$$
a = \sqrt{\frac{1 - \cos l}{2}}
\;\;\;\;\;\;\;\;
z = \sqrt{\sin^2 l - \frac{2}{3}(1 - \cos l)}
$$
When projected into RF space the tetrahedron has vertices:
$$
\begin{array}{ll}
v_1 = \left \{ 0, 0,  0  \right \}
&v_2 = \frac{1}{\cos l} \left \{ -a,  \frac{-a}{\sqrt{3}}, z  \right \}\\
v_3 = \frac{1}{\cos l} \left \{  a,  \frac{-a}{\sqrt{3}}, z  \right \}
&v_4 = \frac{1}{\cos l} \left \{  0, \frac{-2a}{\sqrt{3}}, z  \right \}
\end{array}
$$
The dihedral angle of the tetrahedron is then given by:
$$
\psi \left( l \right) = \arccos \left \langle \frac{v_2 \times v_3}{|v_2 \times v_3|}, \frac{v_2 \times v_4}{|v_2 \times v_4|} \right \rangle
= \arccos \left( \frac{\cos l}{2 \cos l + 1} \right)
$$
Using Equation~(\ref{eq:radius_to_edge_length}) we can express the dihedral angle in terms of \radius:
$$
\psi(\radius) = \arccos\left( \frac{4\cos^2(\radius) - 1}{8\cos^2(\radius) + 1} \right)
$$
The solid angle is then given by:
$$
\Omega(\radius) = 3\psi(\radius) - \pi = 3 \arccos\left( \frac{4\cos^2(\radius) - 1}{8\cos^2(\radius) + 1} \right) - \pi
$$
Since $v_1$ lies at the origin, this is also the solid angle of intersection of a regular hyperspherical simplex and a hyperspherical cap placed at one of its vertices.  We can verify that in the Euclidean limit (where the curvature is zero),
$\lim_{\radius \rightarrow 0} \Omega(\radius) = 3 \arccos\left( \frac{1}{3} \right) - \pi = \arccos \left( \frac{23}{27} \right)$, which is the solid angle for a regular tetrahedron in \vectorspace{R}{3}, 
and that
$
\Omega \left( \frac{\pi}{3} \right) = \frac{\pi}{2}
$
which is the solid angle of a tetrahedral cell in the 16-cell.

\section{Laue Group Subset Relationships}
\label{sec:laue_tables}

The subset relationships between the 11 Laue groups are shown in Tables~\ref{table:cubic_generators} and \ref{table:hexagonal_generators}.

\begin{table}%[h]
\label{table:cubic_generators}
\hcaption{Generators for the seven Laue groups which are subsets of $O$.}
\begin{tabular}{lccrccccccc}
%\hline
 & $O$ & $T$ & $D_4$ & $D_2$ & $C_4$ & $C_2$ & $C_1$\\
\hline
$\left\{ 1, 0, 0, 0 \right\}$ & \cmark & \cmark & \cmark & \cmark & \cmark & \cmark & \cmark\\
$\left\{ 0, 0, 0, 1 \right\}$ & \cmark & \cmark & \cmark & \cmark & \cmark & \cmark & \\
$\left\{ 0, 1, 0, 0 \right\}$ & \cmark & \cmark & \cmark & \cmark &  &  & \\
$\left\{ 0, 0, 1, 0 \right\}$ & \cmark & \cmark & \cmark & \cmark &  &  & \\
$\left\{ \frac{\sqrt{2}}{2}, 0, 0, \frac{\sqrt{2}}{2} \right\}$ & \cmark &  & \cmark &  & \cmark &  & \\
$\left\{ \frac{\sqrt{2}}{2}, 0, 0, -\frac{\sqrt{2}}{2} \right\}$ & \cmark &  & \cmark &  & \cmark &  & \\
$\left\{ 0, \frac{\sqrt{2}}{2}, \frac{\sqrt{2}}{2}, 0 \right\}$ & \cmark &  & \cmark &  &  &  & \\
$\left\{ 0, -\frac{\sqrt{2}}{2}, \frac{\sqrt{2}}{2}, 0 \right\}$ & \cmark &  & \cmark &  &  &  & \\
$\left\{ \frac{1}{2}, \frac{1}{2}, -\frac{1}{2}, \frac{1}{2} \right\}$ & \cmark & \cmark &  &  &  &  & \\
$\left\{ \frac{1}{2}, \frac{1}{2}, \frac{1}{2}, -\frac{1}{2} \right\}$ & \cmark & \cmark &  &  &  &  & \\
$\left\{ \frac{1}{2}, \frac{1}{2}, -\frac{1}{2}, -\frac{1}{2} \right\}$ & \cmark & \cmark &  &  &  &  & \\
$\left\{ \frac{1}{2}, -\frac{1}{2}, -\frac{1}{2}, -\frac{1}{2} \right\}$ & \cmark & \cmark &  &  &  &  & \\
$\left\{ \frac{1}{2}, -\frac{1}{2}, \frac{1}{2}, \frac{1}{2} \right\}$ & \cmark & \cmark &  &  &  &  & \\
$\left\{ \frac{1}{2}, -\frac{1}{2}, \frac{1}{2}, -\frac{1}{2} \right\}$ & \cmark & \cmark &  &  &  &  & \\
$\left\{ \frac{1}{2}, -\frac{1}{2}, -\frac{1}{2}, \frac{1}{2} \right\}$ & \cmark & \cmark &  &  &  &  & \\
$\left\{ \frac{1}{2}, \frac{1}{2}, \frac{1}{2}, \frac{1}{2} \right\}$ & \cmark & \cmark &  &  &  &  & \\
$\left\{ \frac{\sqrt{2}}{2}, \frac{\sqrt{2}}{2}, 0, 0 \right\}$ & \cmark &  &  &  &  &  & \\
$\left\{ \frac{\sqrt{2}}{2}, -\frac{\sqrt{2}}{2}, 0, 0 \right\}$ & \cmark &  &  &  &  &  & \\
$\left\{ \frac{\sqrt{2}}{2}, 0, \frac{\sqrt{2}}{2}, 0 \right\}$ & \cmark &  &  &  &  &  & \\
$\left\{ \frac{\sqrt{2}}{2}, 0, -\frac{\sqrt{2}}{2}, 0 \right\}$ & \cmark &  &  &  &  &  & \\
$\left\{ 0, \frac{\sqrt{2}}{2}, 0, \frac{\sqrt{2}}{2} \right\}$ & \cmark &  &  &  &  &  & \\
$\left\{ 0, -\frac{\sqrt{2}}{2}, 0, \frac{\sqrt{2}}{2} \right\}$ & \cmark &  &  &  &  &  & \\
$\left\{ 0, 0, \frac{\sqrt{2}}{2}, \frac{\sqrt{2}}{2} \right\}$ & \cmark &  &  &  &  &  & \\
$\left\{ 0, 0, -\frac{\sqrt{2}}{2}, \frac{\sqrt{2}}{2} \right\}$ & \cmark &  &  &  &  &  & \\
\end{tabular}
\end{table}

\begin{table}%[h]
\hcaption{Generators for the five Laue groups which are subsets of $D_6$.}
\begin{tabular}{lccccccc}
%\hline
 & $D_6$ & $D_3$ & $C_6$ & $C_3$ & $C_1$\\
\hline
$\left\{ 1, 0, 0, 0 \right\}$ & \cmark & \cmark & \cmark & \cmark & \cmark\\
$\left\{ \frac{1}{2}, 0, 0, \frac{\sqrt{3}}{2} \right\}$ & \cmark & \cmark & \cmark & \cmark & \\
$\left\{ \frac{1}{2}, 0, 0, -\frac{\sqrt{3}}{2} \right\}$ & \cmark & \cmark & \cmark & \cmark & \\
$\left\{ 0, 0, 0, 1 \right\}$ & \cmark &  & \cmark &  & \\
$\left\{ \frac{\sqrt{3}}{2}, 0, 0, \frac{1}{2} \right\}$ & \cmark &  & \cmark &  & \\
$\left\{ \frac{\sqrt{3}}{2}, 0, 0, -\frac{1}{2} \right\}$ & \cmark &  & \cmark &  & \\
$\left\{ 0, 1, 0, 0 \right\}$ & \cmark & \cmark &  &  & \\
$\left\{ 0, -\frac{1}{2}, \frac{\sqrt{3}}{2}, 0 \right\}$ & \cmark & \cmark &  &  & \\
$\left\{ 0, \frac{1}{2}, \frac{\sqrt{3}}{2}, 0 \right\}$ & \cmark & \cmark &  &  & \\
$\left\{ 0, \frac{\sqrt{3}}{2}, \frac{1}{2}, 0 \right\}$ & \cmark &  &  &  & \\
$\left\{ 0, -\frac{\sqrt{3}}{2}, \frac{1}{2}, 0 \right\}$ & \cmark &  &  &  & \\
$\left\{ 0, 0, 1, 0 \right\}$ & \cmark &  &  &  & \\
\end{tabular}
\label{table:hexagonal_generators}
\end{table}

     %-------------------------------------------------------------------------
     % The back matter of the paper - acknowledgements and references
     %-------------------------------------------------------------------------

     % Acknowledgements come after the appendices

\ack{Acknowledgements}
The authors thank Thomas J.~Hardin for advice on spherical and hyperspherical harmonics, Farangis Ram for discussions on the cubochoric method, Nanna Wahlberg, Erik B.~Knudsen and Hugh Simons for proofreading and helpful suggestions on the manuscript, and an anonymous referee for several suggestions which have improved this work.

     % References are at the end of the document, between \begin{references}
     % and \end{references} tags. Each reference is in a \reference entry.

%\begin{references}
%\reference{Author, A. \& Author, B. (1984). \emph{Journal} \textbf{Vol}, first page--last page.}
%\end{references}
\referencelist[refs.bib]

     %-------------------------------------------------------------------------
     % TABLES AND FIGURES SHOULD BE INSERTED AFTER THE MAIN BODY OF THE TEXT
     %-------------------------------------------------------------------------

\end{document}